\documentclass[a4paper,UKenglish,cleveref, autoref, thm-restate,bookmarks,unicode,colorlinks=true,pdfencoding=auto]{lipics-v2021}
\usepackage{amsfonts}
\usepackage{amsmath}
\usepackage{amssymb}

\usepackage{float}
\usepackage{tikz}
\usetikzlibrary{intersections, calc, arrows, automata, positioning,shapes, trees}
\usepackage{bbding}
\usepackage{pifont}
\usepackage{pgf, verbatim}

\usepackage{setspace}
\usepackage[linesnumbered,ruled]{algorithm2e}

   \def\@citecolor{blue}%
   \def\@urlcolor{blue}%
   \def\@linkcolor{blue}%
\def\orcidID#1{\smash{\href{http://orcid.org/#1}{\protect\raisebox{-1.25pt}{\protect\includegraphics{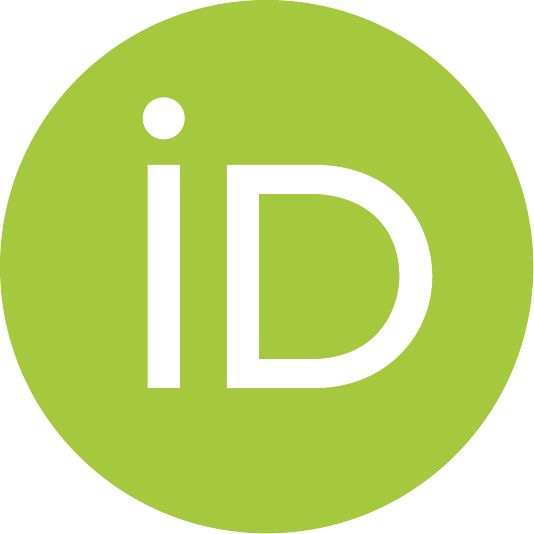}}}}}

\newcommand{\supp}{{\mathsf{support}}}
\newcommand{\pfun}{\mathrel{\ooalign{\hfil$\mapstochar\mkern5mu$\hfil\cr$\to$\cr}}} 
\newcommand{\Dist}{\mathrm{Distr}} 
\newcommand{\SubDist}{\mathrm{SubDistr}} 
\newcommand{\suchthat}{\;\ifnum\currentgrouptype=16 \middle\fi|\;} 
\newcommand{\Paths}{{\rm Paths}} 

\newcommand{\D}{{\mathcal D}} 
\newcommand{\A}{{\mathcal A}}
\newcommand{\B}{{\mathcal B}}
\newcommand{\M}{{\mathcal M}}
\newcommand{\C}{{\mathcal C}} 
\newcommand{\G}{{\mathcal G}} 

\newcommand{\lang}{{\mathcal L}} 
\newcommand{\m}{{\sf m}} 
\newcommand{\Act}{{\mathsf{Act}}} 
\newcommand{\Prob}{{\sf P}} 
\newcommand{\nat}{{\mathbb N}} 

\definecolor{OliveGreen}{rgb}{0,0.6,0}

\pdfoutput=1 
\hideLIPIcs  

\title{On the Succinctness of Good-for-MDPs Automata} 

\author{Sven Schewe}{University of Liverpool, UK}{sven.schewe@liverpool.ac.uk}{https://orcid.org/0000-0002-9093-9518}{}

\author{Qiyi Tang}{University of Liverpool, UK}{qiyi.tang@liverpool.ac.uk}{https://orcid.org/0000-0002-9265-3011}{}

\authorrunning{Sven Schewe and Qiyi Tang} 

\Copyright{Sven Schewe and Qiyi Tang} 



\nolinenumbers 

\EventEditors{John Q. Open and Joan R. Access}
\EventNoEds{2}
\EventLongTitle{42nd Conference on Very Important Topics (CVIT 2016)}
\EventShortTitle{CVIT 2016}
\EventAcronym{CVIT}
\EventYear{2016}
\EventDate{December 24--27, 2016}
\EventLocation{Little Whinging, United Kingdom}
\EventLogo{}
\SeriesVolume{42}
\ArticleNo{23}

\begin{document}\sloppy

\maketitle

\begin{abstract}
Good-for-MDPs and good-for-games automata are two recent classes of nondeterministic automata that reside between general nondeterministic and deterministic automata.
Deterministic automata are good-for-games, and good-for-games automata are good-for-MDPs, but not vice versa.
One of the question this raises is how these classes relate in terms of succinctness.
Good-for-games automata are known to be exponentially more succinct than deterministic automata, but the gap between good-for-MDPs and good-for-games automata as well as the gap between ordinary nondeterministic automata and those that are good-for-MDPs have been open.
We establish that these gaps are exponential, and sharpen this result by showing that the latter gap remains exponential when restricting the nondeterministic automata to separating safety or unambiguous reachability automata.
\end{abstract}

\section{Introduction}
Good-for-games (GfG) and good-for-MDPs (GfM) automata are automata with a special form of nondeterminism that can be resolved on-the-fly.
In the case of good-for-games automata this needs to be possible in every context, while it only needs to be possible for GfM automata in the context of finite MDPs.
This makes the restrictions posed on GfM automata milder.
One of the implications is that GfM automata are less complex than GfG automata, as GfM B\"uchi automata can recognise all regular languages.
But are they also significantly more succinct?

We show that this is the case, establishing that GfM automata are exponentially more succinct than GfG automata, where exponentially more succinct means that a translation of automata with $n$ states (or transitions) require automata with $2^{\Omega(n)}$ states in the worst case.

To show this, we develop a family $\G_n$ of automata with $n+2$ states and $3n+7$ transitions, 
and show that the smallest language equivalent deterministic B\"uchi automaton (DBAs) has at least $2^n / (n+2)$ states. As GfG B\"uchi automata are only quadratically more succinct than DBAs \cite{KuperbergS15}, this implies that GfG auatomata recognising the language are in $\Omega(\sqrt{2^n/n})$ (or: $\Omega(2^{n/2}/\sqrt{n})$).

We then continue to show a similar gap between B\"uchi automata with general nondeterminism and those that are GfM, producing two families of automata, $\mathcal R_n$ and $\mathcal S_n$, with $n+2$ states and $2n+5$ 
transitions and $n+1$ states and $2n$ transitions, respectively, and show that a language equivalent GfM automaton requires at least $2^n$ states to recognise either of these languages.

This is not only close to the known translation to slim automata \cite{HPSSTWBP2020} that results in $3^n$ states (and no more than $2$ successors for any state letter pair), but it provides more: each $\mathcal R_n$ is an unambiguous%
\footnote{An \emph{unambiguous} automaton has, for every $\omega$-word, no more than one accepting run}
reachability automaton, while each $\mathcal S_n$ is a safety automaton, which is strongly unambiguous%
\footnote{A \emph{strongly unambiguous} automaton has, for every $\omega$-word, no more than one run}
and separating%
\footnote{An automaton is \emph{separated} if different states have disjoint languages; separating automata are therefore in particular unambiguous.}. 
Moreover, we show the languages recognised by $\mathcal R_n$ are also recognised by a separating B\"uchi automaton $\mathcal R_n'$ with $n+2$ states and $3(n+2)$ transitions.

That the succinctness gap extends to these smaller classes of automata is not only interesting from a principle point of view, unambiguous automata appear naturally both as a standard translation from LTL, and in model checking Markov chains \cite{BaierKKMW23}.
Being unambiguous is thus a different type of restriction to nondeterminism, very different in nature, and yet with related fields of application to GfM.
It is also interesting to note that strong disambiguation is more expensive ($\Omega((n-1)!)$, \cite{KarmarkarJC13}) than translating a nondeterministic B\"uchi automaton into one that is GfM ($O(3^n)$, cf.\ \cite{Courco90,Hahn15,Sicker16b,HPSSTWBP2020}).

\section{Preliminaries}
We write $\nat$ for the set of nonnegative integers. 
Let $S$ be a finite set. 
We denote by $\Dist(S)$ ($\SubDist(S)$) the set of probability (sub)distributions on~$S$. 
For a (sub)distribution $\varphi \in \Dist(S)$ (resp.\ $ \mu \in \SubDist(S)$), we write $\supp(\varphi) = \{s \in S \suchthat \varphi(s) > 0 \}$ for its support. 
The cardinal of $S$ is denoted $|S|$. 
We use $\Sigma$ to denote a finite alphabet. 
We denote by $\Sigma^{*}$ the set of finite words over $\Sigma$ and $\Sigma^{\omega}$ the set of $\omega$-words over $\Sigma$. 
We use the standard notions of prefix and suffix of a word. 
By $w\alpha$ we denote the concatenation of a finite word $w$ and an $\omega$-word $\alpha$. 

\subsection{Automata}\label{subsection:preliminaries-automata} 

A \emph{nondeterministic} 
word automaton 
is a tuple $\A= (\Sigma, Q, q_0, \delta, F)$, where $\Sigma$ is a finite alphabet, $Q$ is a finite set of states, $q_0 \in Q$ is the initial state, $\delta: Q \times \Sigma \to 2^Q$ is the transition function, and $F \subseteq Q$ is the set of final (or accepting) states%
\footnote{Especially for GfG automata, the decision of putting acceptance on states or transitions can be quite important for some questions, e.g.\ the minimisation of CoB\"uchi GfG automata is tractable for transition based acceptance~\cite{RadiK22}, but not for state based acceptance \cite{Schewe20}. For establishing exponential succinctness gaps, however, it is not relevant, as the translation from the more succinct transition based to state based acceptance can only double the size of an automaton.}. 
A nondeterministic word automaton is \emph{complete} if $\delta(q, \sigma) \neq \emptyset$ for all $q \in Q$ and $\sigma \in \Sigma$.
For a state $q \in Q$, we denote $\A_q$ the NBA obtained from $\A$ by making $q$ the initial state.

A finite run $r$ of $\A$ on $w \in \Sigma^{*}$ is a finite word $q_0, w_0, q_1, w_1, \ldots \in Q \times (\Sigma  \times Q)^{*}$ such that $q_i \in \delta(q_{i-1} , w_{i-1})$ for all $i > 0$.
A nondeterministic finite automaton (NFA) $\A$ accepts exactly those runs, in which it ends in a state in~$F$.
A word in $\Sigma^*$ is accepted by the automaton if it has an accepting run, and the language of an automaton, denoted $\lang(\A)$, is the set of accepted words in $\Sigma^{*}$.  
A family of NBAs is given in \cref{fig:NFA}.

A run $r$ of $\A$ on $w \in \Sigma^{\omega}$ is an $\omega$-word $q_0, w_0, q_1, w_1, \ldots \in (Q \times \Sigma)^{\omega}$ such that $q_i \in \delta(q_{i-1} , w_{i-1})$ for all $i > 0$.
A nondeterministic B\"uchi word automaton (NBA) $\A$ accepts exactly those runs, in which at least one of the infinitely often occurring states is in~$F$. 
A word in $\Sigma^\omega$ is accepted by the automaton if it has an accepting run, and the language of an automaton, denoted $\lang(\A)$, is the set of accepted words in $\Sigma^{\omega}$.  
Two example NBAs are given in \cref{fig:GFM-G1,fig:NSA-S1}.

A nondeterministic word automaton is called a \emph{safety automaton} if all states are final (and therefore all runs accepting), and a \emph{reachability automaton} if there is only a single final state, say $f$, and that final state is a sink, i.e.\ $\delta(f,\sigma) = \{f\}$ for all $\sigma \in \Sigma$.

\begin{figure}[t]
\centering
\begin{tikzpicture} [->, node distance = 2cm, auto,initial text = {}]
    \node[initial left,state] (q0) at (0,1) {$q_0$};
    \node[state] (q1) at (2,1) {$q_1$};
    \node[state, accepting] (qf) at (4,1) {$f$};

    \path (q0) edge [loop above] node {$0,1,\$$} (q0);
    \path (q0) edge [bend left] node [above] {$1$} (q1);
    \path (q1) edge node [below]{$\$$} (qf);
    \path (q1) edge [bend left] node [above] {$0,1$} (q0);
    \path (qf) edge [bend left=35] node [below] {$0,1,\$$}  (q0);

\end{tikzpicture}
\caption{The NBA $\G_1$ which recognises the language $\lang_1^{\omega}$ in which there are infinitely many finite words from $\lang_1 = \{0, 1\}^{*}1\$$. This NBA is GfM. It is neither GfG nor unambiguous (or separated).}
\label{fig:GFM-G1}
\end{figure}
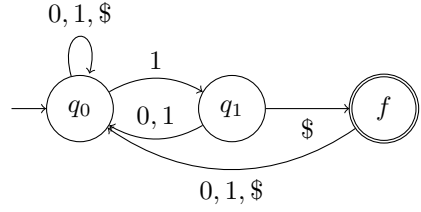

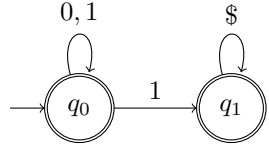
\begin{figure}[ht]
\centering
\begin{tikzpicture} [->, node distance = 2cm, auto,initial text = {}]
    \node[initial left,state,accepting] (q0) at (0,1) {$q_0$};
    \node[state, accepting] (qf) at (2,1) {$q_1$};

    \path (q0) edge [loop above] node {$0,1$} (q0);
    \path (q0) edge node {$1$} (qf);
    \path (qf) edge [loop above] node {$\$$} (qf);
\end{tikzpicture}
\caption{
The safety NBA $\mathcal S_1$, which recognises the language $\lang_1^s = \{0, 1\}^{*}1\$^{\omega} \cup \{0,1\}^\omega$ that contains all words over the alphabet $\{0,1,\$\}$ that contain either no $\$$, or that contain a $1$ that is succeeded by only $\$$s.
This NBA is separated (and thus unambiguous), but neither GfG nor GfM.
}
\label{fig:NSA-S1}
\end{figure}


A nondeterministic word automaton is \emph{deterministic} if the transition function $\delta$ maps each state and letter pair to a singleton set, a set consisting of a single state. 
For deterministic automata, we also extend $\delta$ to words, by letting $\delta(q, \varepsilon) = q$ and $\delta(q, w_{0} w_{1} \cdots w_{n}) = \delta(\delta(q, w_{0}), w_1 \cdots w_{n})$, where we have $n \geq 1$ and $w_{i} \in \Sigma$ for $i \in \{0, \cdots, n\}$.

A nondeterministic automaton is called \emph{good-for-games (GfG)} if it only relies on a limited form of nondeterminism: GfG automata can make their decision of how to resolve their nondeterministic choices on the history at any point of a run rather than using the knowledge of the complete word as a nondeterministic automaton normally would without changing their language. They can be characterised in many ways, including as automata that simulate deterministic automata.  


The NBA in \cref{fig:GFM-G1} is good-for-MDPs (GfM) as shown later in \cref{lem:gfm}, but neither GfG nor unambiguous (or separated). 
The NBA in \cref{fig:NSA-S1} is neither GfG nor good-for-MDPs (GfM) as shown later at the end of \cref{subsection:preliminaries-product-mdp-gfm}, while it is both unambiguous and separated.

\subsection{Markov Decision Processes (MDPs)}\label{subsection:preliminaries-mdp}

A \emph{(finite, transition-labelled) Markov decision process} (MDP) is a tuple $\langle S, \Act, \Prob, \Sigma , \ell\rangle$ consisting of a finite set $S$ of states, a finite set $\Act$ of actions,
a partial function 
$\Prob: S \times \Act \times \Sigma \pfun \SubDist(S)$ 
denoting the probabilistic transition, and
a labelling function $\ell: S \times \Act \times \Sigma \times  S \to \Sigma$ such that $\ell(s, \m , \sigma, s') = \sigma$.
We have $\sum_{\sigma \text{ s.t. } \Prob(s, \m, \sigma) \text{ is defined} }\Prob(s, \m, \sigma) \le 1$. 
The set of available actions in a state $s$ is $\Act(s) = \{\m \in \Act \suchthat \exists \sigma \in \Sigma \text{ such that } \Prob(s, \m, \sigma) \text{ is defined}\}$. 
An MDP is a Markov chain (MC) if $|\Act(s)| = 1$ for all $s \in S$. 

\begin{figure}[ht]
\centering
\begin{tikzpicture} [->, node distance = 2cm, auto,initial text = {}]
    \node[initial left,state] (q0) at (0,1) {$s_0$};
    \node[state] (qn) at (2,1) {$s_1$};
    \node[state] (qf) at (4,1) {$s_f$};

    \path (q0) edge node {$\textcolor{red}{1}:0$} (qn);
    \path (qn) edge [loop above] node {$\textcolor{red}{\frac{1}{4}}: 0$} (qn);
    \path (qn) edge [loop below] node {$\textcolor{red}{\frac{1}{4}}: 1$} (qn);
    \path (qn) edge node {$\textcolor{red}{\frac{1}{2}}: \$$} (qf);
    \path (qf) edge [loop above] node {$\textcolor{red}{1}:\$$} (qf);

\end{tikzpicture}
\caption{
An MDP with initial state $s_0$. Since each state has only one available action $\m$ (not shown in the figure), the MDP is actually an MC. 
The set of labels is $\{0, 1, \$\}$ and the labelling function for the MC is as follows: $\ell(s_0, \m, 0, s_1) = \ell(s_1, \m, 0, s_1) = 0$, $\ell(s_1, \m, 1,s_1) = 1$, $\ell(s_1, \m, \$, s_f) = \ell(s_f, \m, \$, s_f) = \$$. 
The probabilities of the transitions are indicated in red before the labels. 
The chance that it produces a word in $\lang_1^s$ recognised by $\mathcal S_1$ from \cref{fig:NSA-S1} is $\frac{1}{4}$.
}
\label{fig:example-MC}
\end{figure}
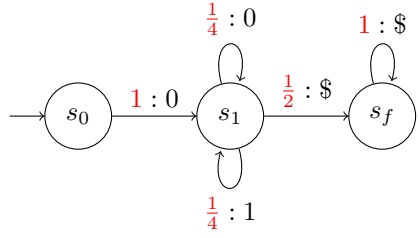


An infinite \emph{run (path)} of an MDP $\M$ is a sequence $s_0\m_1\sigma_1\ldots \in (S \times \Act \times \Sigma)^\omega$ such that $\Prob(s_i, \m_{i+1}, \sigma_{i+1})$ is defined and $\Prob(s_i, \m_{i+1}, \sigma_{i+1})(s_{i+1}) > 0$ for all $i \ge 0$.
A finite run is a finite such sequence.
Let $\Omega(\M)$ ($\Paths(\M)$) denote the set of (finite) runs in $\M$ and $\Omega(\M)_s$ ($\Paths(\M)_s$) denote the set of (finite) runs in $\M$ starting from $s$.
Abusing the notation slightly, for an infinite run $r = s_0\m_1\sigma_1s_1\m_2\sigma_2\ldots$ we write $\ell(r) = \sigma_1\sigma_2\ldots \in \Sigma^{\omega}$.

A \emph{strategy} for an MDP is a function $\mu: \Paths(\M) \to \Dist(\Act)$ that, given a finite run~$r$, returns a probability distribution on all the available actions at the last state of $r$. 
A \emph{deterministic} strategy for an MDP is a function $\mu: \Paths(\M) \to \Act$ that, given a finite run~$r$, returns a single available action at that state. 
The set of runs $\Omega(\M)_s^{\mu}$ is a subset of $\Omega(\M)_s$ that correspond to strategy $\mu$ and initial state $s$.
Given a finite-memory strategy $\mu$ for $\M$, an MC $(\M)_{\mu}$ is induced \cite[Section~10.6]{Baier08}.


A \emph{sub-MDP} of $\M$ is an MDP $\M' = \langle S', \Act', \Prob', \Sigma , \ell' \rangle$, where $S' \subseteq S$, $\Act' \subseteq \Act$ is such that $\Act'(s) \subseteq \Act(s)$ for every $s \in S'$, and $\Prob'$ and $\ell'$ are analogous to $\Prob$ and $\ell$ when restricted to $S'$ and $\Act'$.  In particular, $\M'$ is closed under probabilistic transitions, i.e. for all $s \in S'$, $\m \in \Act'$ and $\sigma \in \Sigma$ we have that $\Prob'(s, \m, \sigma)(s') > 0$ implies that $s' \in S'$. An \emph{end-component} \cite{Alfaro1998,Baier08} of an MDP $\M$ is a sub-MDP $\M'$ of $\M$ such that its underlying graph is strongly connected\footnote{A set of states is strongly connected if every pair of states are mutually reachable.} and it has no outgoing transitions. An example MDP (MC) is presented in Figure~\ref{fig:example-MC}.

A strategy $\mu$ and an initial state $s \in S$ induce a standard probability measure on sets of infinite runs, see, e.g., \cite{Baier08}. Such measurable sets of infinite runs are called events or objectives. We write $\mathrm{Pr}_s^{\mu}$ for the probability of an event $E \subseteq sS^{\omega}$ of runs starting from $s$.


\subsection{The Product of MDPs and NBAs}\label{subsection:preliminaries-product-mdp-gfm}

Given an MDP $\M = \langle S, \Act, \Prob, \Sigma, \ell \rangle$ with initial state $s_0 \in S$ and an NBA $\A = \langle \Sigma, Q, \delta, q_0, F \rangle$, we want to compute an optimal strategy satisfying the objective that the run of $\M$ is in the language of $\A$. We define the semantic satisfaction probability for $\A$ and a strategy $\mu$ from state $s \in S$ as:
$\mathrm{PSem}_{\A}^{\M}(s, \mu) = \mathrm{Pr}_s^{\mu}\{ r \in \Omega_s^{\mu}: \ell(r) \in \lang(\A)\}$ and $\mathrm{PSem}_{\A}^{\M}(s) = \sup_{\mu} \big( \mathrm{PSem}_{\A}^{\M}(s, \mu) \big)$. In the case that $\M$ is an MC, we simply have $\mathrm{PSem}_{\A}^{\M}(s) = \mathrm{Pr}_s\{ r \in \Omega_s: \ell(r) \in \lang(\A)\}$.

\begin{figure*}[t]
\centering
\begin{tikzpicture} [xscale=.7,shorten >=1pt,node distance = 1cm, on grid, auto,initial text = {}]
    \usetikzlibrary{positioning,arrows,automata}
    \tikzstyle{BoxStyle} = [draw, circle, fill=black, scale=0.4,minimum width = 1pt, minimum height = 1pt]
    \node[initial left,state, accepting] (s0q0) at (-5,0) {$s_0, q_0$};
    \node[BoxStyle] (act00) at (-3,0) {};
    \node[label] at (-3, 0.35) {$(\m,q_0)$};

    \node[state, accepting] (s1q0) at (0,0) {$s_1, q_0$};
    \node[BoxStyle] (act1) at (2,0) {};
    \node[label] at (3,0) {$(\m,q_0)$};
    \node[BoxStyle] (act2) at (0,-1.2) {};
    \node[label] at (-1,-1.2) {$(\m, q_1)$};
    
    \node[state, accepting] (s1q1) at (0,-3) {$s_1, q_1$};
    \node[BoxStyle] (act11) at (2,-3) {};
    \node[label] at (2,-2.65) {$(\m,q_2)$};
    
    \node[state, accepting] (sfq1) at (5,-3) {$s_f, q_1$};				
    \node[BoxStyle] (actf1) at (6.5,-3) {};
    \node[label] at (7.5,-3) {$(\m,q_1)$};
    \node[label] at (6.5,-3.5) {$\textcolor{red}{1}:\$$};

    \path [-] (s0q0) edge [above] node {} (act00);
    \path [->] (act00) edge [below, midway] node {$\textcolor{red}{1}:0$} (s1q0);

    \path [-] (s1q0) edge [above] node {} (act1);
    \path [-] (s1q0) edge [above] node {} (act2);
    \path [->] 
    (act1) edge [bend left, below] node {$\textcolor{red}{\frac{1}{4}}:1$} (s1q0)
    (act1) edge [bend right, above] node {$\textcolor{red}{\frac{1}{4}}:0$} (s1q0);
    \path [->] (act2) edge [left] node {$\textcolor{red}{\frac{1}{4}}:1$} (s1q1);

    \path [-] (s1q1) edge [above] node {} (act11);
    \path [->] (act11) edge [below] node {$\textcolor{red}{\frac{1}{2}}:\$$} (sfq1);

    \path [->] 
    (sfq1) edge [loop right, near end] node {} (sfq1);
    
\end{tikzpicture}
\caption{An example of a product MDP $\M \times {\mathcal S}_1$ with initial state $(s_0, q_0)$ and $F^{\times} = \{(s_1, q_1), (s_f, q_1)\}$ where $\M$ is the MDP (MC) in \cref{fig:example-MC} and ${\mathcal S}_1$ is the safety NBA from \cref{fig:NSA-S1}. The maximum chance for the product MDP of generating a word in $\lang(\mathcal{S}_1) = \lang_n^s$ is $\frac{1}{8}$. }
\label{fig:product-MDP}
\end{figure*}
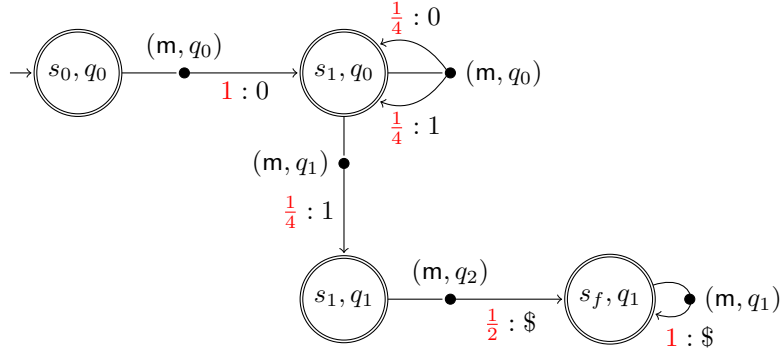

The \emph{product} of $\M$ and $\A$ is an MDP $\M \times \A= \langle S \times Q, \Act \times Q, \Prob^{\times}, \Sigma, \ell^{\times} \rangle$ augmented with the initial state $(s_0, q_0)$ and the B\"uchi acceptance condition $F^{\times}
= \{(s, q) \in S \times Q \suchthat q \in F \}$. The labelling function $\ell^{\times}$ maps each transition $\big( (s, q), (\m, q'), (s',q') \big)$ to $\ell\big((s, \m, s')\big)$. 

We define the partial function $\Prob^{\times}: (S \times Q) \times (\Act \times Q) \pfun \Dist(S \times Q)$ as follows: for all $(s, \m) \in \supp(\Prob)$
, $s' \in S$ and $q, q' \in Q$,  we have  
$$\Prob^{\times}\big((s, q), (\m, q'), \sigma \big) \big( (s', q') \big) = 
\left \{ \begin{array}{ll}
    \Prob(s, \m, \sigma)(s') & \mbox{if $q' \in \delta(q, \sigma)$}\\
    0 & \mbox{otherwise.}
    \end{array}
\right .
$$


We define the syntactic satisfaction probability for the product MDP and a strategy $\mu^{\times}$ from a state $(s, q)$ as:
$\mathrm{PSyn}_{\A}^{\M} \big( (s, q), \mu^{\times}\big) = \mathrm{Pr}_{s, q}^{\mu^{\times}}\{ r \in \Omega_{s, q}^{\mu^{\times}}: \ell^{\times}(r) \in \lang(\A) \}$\footnote{Let $\mathrm{inf} (r)$ be the set of states that appears infinite often in a run $r$. We also have $\mathrm{PSyn}_{\A}^{\M}( (s, q), \mu^{\times}) = \mathrm{Pr}_{s, q}^{\mu^{\times}}\{ r \in \Omega_{s, q}^{\mu^{\times}}: \mathrm{inf} (r) \cap F^{\times} \neq \emptyset\}$.} and $\mathrm{PSyn}_{\A}^{\M}(s) = \sup_{\mu^{\times}} (\mathrm{PSyn}_{\A}^{\M}\big( (s,q_0), \mu^{\times})\big)$. 
The set of actions is $\Act$ in the MDP $\M$ while it is $\Act \times Q$ in the product MDP. This makes $\mathrm{PSem}$ and $\mathrm{PSyn}$ potentially different. In general, $\mathrm{PSyn}_{\A}^{\M}(s) \le \mathrm{PSem}_{\A}^{\M}(s)$ for all $s \in S$, because accepting runs can only occur on accepted words. 
If $\A$ is deterministic, $\mathrm{PSyn}_{\A}^{\M}(s) = \mathrm{PSem}_{\A}^{\M}(s)$ holds for all $s \in S$. 

End-components and runs are defined for products just like for MDPs. A run of $\M \times \A$ is accepting if it satisfies the product’s acceptance condition. An accepting end-component of $\M \times \A$ is an end-component which contains some states in $F^{\times}$. 

\begin{definition}
    \cite{HPSSTWBP2020}
    \label{def:GFM}
    An NBA $\A$ is good-for-MDPs (GfM) if, for all finite MDPs $\M$ with initial state $s_0$, $\mathrm{PSyn}_{\A}^{\M}(s_0) = \mathrm{PSem}_{\A}^{\M}(s_0)$ holds.  
\end{definition}

An example of a product MDP is presented in \cref{fig:product-MDP}. It is the product of the MC in Figure~\ref{fig:example-MC} and the safety NBA in \cref{fig:NSA-S1}.
Note that we have omitted rejecting sink states, so that some outgoing probabilities do not sum up to $1$.
The chance that the MC produces a word in $\lang_1^s = \lang(\mathcal{S}_1)$ recognised by the automaton from \cref{fig:NSA-S1} is $\frac{1}{4}$. However, the syntactic satisfaction probability $\mathrm{PSyn}_{\mathcal{S}_1}^{\M}(s_0) = \frac{1}{8}$ is witnessed by the memoryless strategy that chooses the action $(\m,q_1)$ at the state $(s_1,q_0)$. We do not need to specify the strategy for the other states since there is only one available action for any remaining state. According \cref{def:GFM}, the reachability NBA $\mathcal{S}_1$ in \cref{fig:NSA-S1} is not GfM as witnessed by the MC in Figure~\ref{fig:example-MC}. 

\section{Outline}

To establish our main result that
\begin{quote}
GfM automata are exponentially more succinct than GfG automata,
\end{quote}
we proceed in three steps.
We first recall the standard result that NFAs are exponentially more succinct than deterministic finite automata (DFAs), and use a textbook example of a family of languages that witnesses this, namely
\[\lang_n = \{0, 1\}^{*}1\{0, 1\}^{n-1}\$\ ,\]
the language of those words $w \in \{0, 1\}^{*}\$$, where the $n^{\mbox{\scriptsize th}}$ letter before $\$$ in $w$ is $1$.
NFAs that recognise this language need $n+2$ states, while DFAs need $>2^n$.
\medskip

In a second step, we then build an NBA from the standard NFA for this language that recognises the $\omega$-language of infinite words that contain infinitely many sequences of the form $1\, \{0,1\}^{n-1} \$$,
such that the resulting automaton has the same states as the NFA, and almost the same structure (cf.\ \cref{fig:NFA,fig:NBA}).
We call this language $\lang_n^\omega$ and show that the resulting NBA that recognises it is GfM.

This $\omega$-language proves to have many useful properties:
\begin{itemize}
    \item it is B\"uchi recognisable (i.e.\ recognised by a \emph{deterministic} B\"uchi automaton) and
    \item it is a tail language, such that all reachable states of a DBA recognising it are language equivalent and complete.
\end{itemize}

It is indeed easy to obtain a DBA that recognises $\lang_n^\omega$ from a DFA that recognises $\lang_n$ in the same way as we obtained a (GfM) NBA for the NFA that recognises $\lang_n$.
\smallskip

For our third and final step we observe that, should the resulting DBA be minimal---or should a minimal DBA be broadly of the same size---then we are essentially done: a minimal DBA would be exponential in the size of the minimal GfM NBA, and using that GfG NBA are only quadratically more succinct than DBAs \cite{KuperbergS15}, our main result would follow.

The heavy lifting we need to do is therefore to establish just this: that a minimal DBA that recognises $\lang_n^\omega$ cannot be significantly smaller than a minimal DFA that recognises $\lang_n$.
\bigskip

After having established our main result, we also show that
    general nondeterministic automata are exponentially more succinct than good-for-MDP automata.
This even holds when the languages are restricted to reachability or safety languages, and indeed if the automata are restricted to unambiguous reachability automata or separating safety automata.

The languages we use to demonstrate this are the extension of $\lang_n$ to reachability,
\[\lang_{n}^r = \{w \{0,1,\$\}^{\omega} \mid w \in \lang_n\}\ .\]
Instead of using its safety closure,
\[\lang_{n}' = \lang_n^r \cup \{0, 1\}^{\omega}\ , \]
we use the safety language
\[\lang_{n}^s = \{w \$^{\omega} \mid w \in \lang_n\} \cup  \{0, 1\}^{\omega}\ , \]
because this translates more natural into a separating automaton.

We show that $\lang_n^r$ resp.\ $\lang_n^s$ are recognised by a nondeterministic reachability resp.\ safety automaton with $n+2$ states, while we show that a GfM automaton for these languages needs to have $2^n$ states, not counting accepting or rejecting sinks.

We then obtain the refined results by showing that the safety automaton is separating, while the reachability automaton is unambiguous.
We finally show that the reachability language $\lang_n^r$ can be recognised by a separating NBA $\mathcal R_n'$; however, this automaton is not, and cannot be, a reachability automaton.

\section{GfM vs GfG}

\subsection{The language $\lang_n^{\omega}$ and the GfM NBA $\G_n$}
 
For an arbitrary $n \in \nat$, define the language
$$\lang_n = \{w \in \{0, 1\}^{*}\$ \mid \text { the $n^{\mbox{\scriptsize th}}$ letter before $\$$ in $w$ is $1$}\}\ ,$$
where we assume for clarity that all words of length $\leq n$ are not in $\lang_n$. 

\begin{claim}
For any $n$, there is an NFA with $n+2$ states recognising $\lang_{n}$. 
\end{claim}
Let $\Sigma = \{0, 1, \$\}$. Define the NFA $\A_n = (\Sigma, Q, q_0, \delta_\A, \{f\})$ where $Q = \{q_0, \ldots, q_n, f\}$ and  $\delta_\A$ is defined as follows:
$\delta_\A (q_0, 0) = \{q_0\}$; 
$\delta_\A (q_0, 1) = \{q_0, q_1\}$; 
$\delta_\A (q_i, \sigma) = \{q_{i+1}\}$ for all $i \in \{1, \ldots, n-1\}$ and $\sigma \in \{0, 1\}$;
$\delta_\A (q_n, \$) = \{f\}$.

\begin{figure}[t]
\centering
\begin{tikzpicture} [->, node distance = 2cm, auto,initial text = {}]
    \node[initial left,state] (q0) at (0,1) {$q_0$};
    \node[state] (q1) at (2,1) {$q_1$};
    \node[state] (q2) at (4,1) {$q_2$};
    \node[state] (q3) at (6,1) {$q_3$};
    \node (qi) at (7,1) {$\ldots$};
    \node[state] (qn) at (8,1) {$q_n$};
    \node[state, accepting] (qf) at (10,1) {$f$};

    \path (q0) edge [loop above] node {$0,1$} (q0);
    \path (q0) edge node {$1$} (q1);
    \path (q1) edge node {$0,1$} (q2);
    \path (q2) edge node {$0,1$} (q3);
    \path (qn) edge node {$\$$} (qf);
\end{tikzpicture}
\caption{The NFA $\A_n$ that recognises $\lang_n$.}
\label{fig:NFA}
\end{figure}
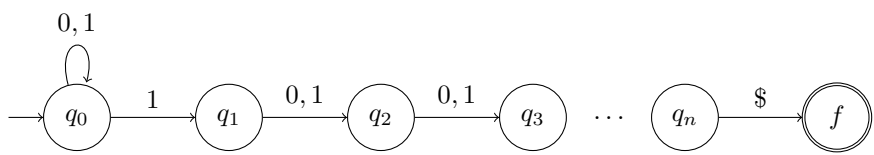

\begin{figure}[t]
\centering
\begin{tikzpicture} [->, node distance = 2cm, auto,initial text = {}]
    \node[initial left,state] (q0) at (0,1) {$q_0$};
    \node[state] (q1) at (2,1) {$q_1$};
    \node[state] (q2) at (4,1) {$q_2$};
    \node[state] (q3) at (6,1) {$q_3$};
    \node (qi) at (7,1) {$\ldots$};
    \node[state] (qn) at (8,1) {$q_n$};
    \node[state, accepting] (qf) at (10,1) {$f$};

    \path (q0) edge [loop above] node {$0,1,\$$} (q0);
    \path (q0) edge [bend left] node [above] {$1$} (q1);
    \path (q1) edge [bend left] node [below]{$0,1$} (q2);
    \path (q1) edge [bend left] node [above] {$\$$} (q0);
    \path (q2) edge [bend left] node [below] {$0,1$} (q3);
    \path (q2) edge [bend left] node [below, near start, yshift=0.2em]{$\$$} (q0);
    \path (q3) edge [bend left]  node [below, near start, yshift=0.2em]{$\$$} (q0);
    \path (qn) edge node {$\$$} (qf);
    \path (qn) edge [bend left] node [below, near start, yshift=0.1em] {$0,1$}  (q0);
    \path (qf) edge [bend left] node [below, near start] {$0,1,\$$}  (q0);

\end{tikzpicture}
\caption{The NBA $\G_n$ that recognises $\lang_n^{\omega}$.}
\label{fig:NBA}
\end{figure}
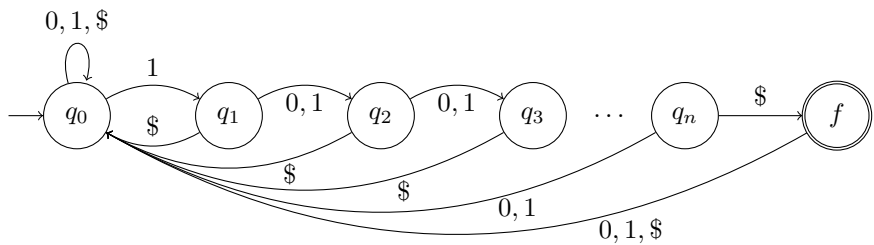


Since it is a standard textbook language, we give the following theorem without proof. 
The intuition is that a DFA needs to store the previous $n$ letters in their order of occurrence in order to recognise this language.

\begin{theorem}\label{theorem:DFA}
For all $n$, any DFA recognising $\lang_{n}$ must have at least $2^n$ states. 
\end{theorem}

\begin{claim}
For any $n$, there is an NBA with $n+2$ states recognising $\lang_{n}^\omega$. 
\end{claim}
Define the NBA $\G_n = (\Sigma, Q, q_0, \delta, \{f\})$ where $Q = \{q_0, \ldots, q_n, f\}$ and  $\delta$ is defined as follows:
$\delta (q_0, 0) = \{q_0\}$; 
$\delta (q_0, 1) = \{q_0, q_1\}$; 
$\delta (q_i, \sigma) = \{q_{i+1}\}$ for all $i \in \{1, \ldots, n-1\}$ and $\sigma \in \{0, 1\}$;
$\delta (q_n, \$) = \{f\}$;
$\delta (q_n, \sigma) = \{q_0\}$ for $\sigma \in \{0, 1\}$; 
$\delta (f, 0) = \{q_0\}$; 
$\delta (f, 1) = \{q_0\}$; 
$\delta (q, \$) = \{q_0\}$ for all $q_n \neq q \in Q$.

$\G_n$ is obtained from $\A_n$ by 
returning to $q_0$ whenever the resulting automaton would block.
The intuition why $\G_n$ is GfM is that, if there is any sequence in $\lang_n$ in an BSCC\footnote{A strongly connected component (SCC) is a maximally strongly connected set of states. 
A bottom SCC (BSCC) is an SCC from which no state outside it is reachable.} of a Markov chain, then it can almost surely try infinitely often to facilitate it for acceptance.
Each of these attempts will succeed with a positive probability $p>0$, so that almost surely infinitely many of these infinite attempts will succeed.

\begin{lemma}\label{lem:gn_correct}
$\G_n$ recognises $\lang_n^\omega$.
\end{lemma}
\begin{proof}
\textbf{`$\lang(\G_n) \subseteq \lang_n^\omega$':} 
It is easy to see that by the construction of $\G_n$, to visit $f$ infinitely often, $\G_n$ has to take the path $q_0, 1, q_1, \sigma_1, q_2, \ldots, q_{n}, \$, f$ on the word $1\sigma_1\cdots\sigma_{n-1}\$ \in \lang_n$ where $\sigma_i \in \{0,1\}$ for all $i \in \{1, \ldots, n-1\}$ infinitely often.

\textbf{`$\lang_n^\omega \subseteq \lang(\G_n)$':} 
We only need to show there is an accepting run of $\G_n$ on any word in $\lang_n^\omega$.
A word $w \in \lang_n^\omega$ is of the form $w_1w_1'w_2w_2'w_3w_3'\ldots$ such that, for all $i\in \nat$, $w_i \in \{0, 1, \$\}^*$ and $w_i' \in \lang_n$.
Let $w_i' = u_i1v_i$ where $u_i \in \{0, 1\}^*$ and $v_i \in \{0, 1\}^{n-1}\$$ for all $i > 0$. We can build the run as follows.

For odd indices $i$, $\G_n$ always stays in $q_0$ on $w_iu_i$, transitions to $q_1$ on $1$ and to $f$ on $v_i$;
(at the beginning, $\G_n$ starts from $q_0$ and stays in $q_0$ on $w_1u_1$, transitions to $q_1$ on $1$ and to $f$ on $v_1$.)
For even indices $i$, $\G_n$ starts from the state $f$, then traverses to $q_0$ and stays in $q_0$ on $w_iw_i'$. 


For every $w_i' \in \lang_n$ such that $i$ is odd, the run we build will visit the accepting state $f$. 
As there are infinitely many such $w_i'$s, the run will visit the accepting state $f$ infinitely often, and thus, is accepting.  
\end{proof}

\begin{lemma}\label{lem:gfm}
$\G_n$ is good for MDPs.
\end{lemma}

\begin{proof}

Consider an arbitrary MC $\M$ with initial state $s_0$. We show that $\G_n$ is good for $\M$, that is, $\mathrm{PSem}_{\G_n}^{\M}(s_0) = \mathrm{PSyn}_{\G_n}^{\M}(s_0)$. It suffices to show $\mathrm{PSyn}_{\G_n}^{\M}(s_0) \ge \mathrm{PSem}_{\G_n}^{\M}(s_0)$ since by definition the converse $\mathrm{PSem}_{\G_n}^{\M}(s_0) \ge \mathrm{PSyn}_{\G_n}^{\M}(s_0)$ always holds.

First, we construct a language equivalent deterministic B\"uchi automaton (DBA) $\D_n$. 
We start with the NFA $\A_n'$, see \cref{fig:NFA2}, which recognises the language $\{w \in \{0, 1\}^{*} \mid \text { the $n^{\mbox{\scriptsize th}}$ to the last letter in $w$ is $1$}\}$. 
The NFA $\A_n'$ is very similar to $\A_n$: it has the same states $q_i$ for $i \in \{0, \ldots, n\}$ and the same transitions between them; 
the only difference is that $q_n$ is accepting without any outgoing transitions as there is no $f$ state in $\A_n'$. 
We then determinise $\A_n'$ to a DFA $\B_n'$ by a standard subset construction and then obtain $\B_n$ from $\B_n'$ by adding a fresh $f$ state, adding a $\$$ transition from all accepting states to $f$ and making $f$ the only accepting state. 
It is easy to see that $\B_n$ is deterministic and $\lang(\A_n) = \lang(\B_n)$.  
Now we obtain the DBA $\D_n$ from the DFA $\B_n$ in the same way as we obtain the NBA $\G_n$ from the NFA $\A_n$.

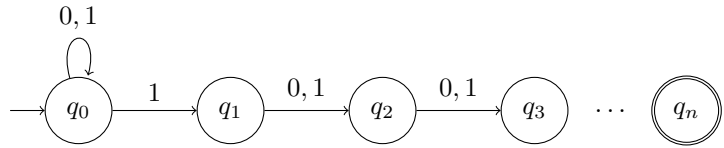
\begin{figure}[ht]
\centering
\begin{tikzpicture} [->, node distance = 2cm, auto,initial text = {}]
    \node[initial left,state] (q0) at (0,1) {$q_0$};
    \node[state] (q1) at (2,1) {$q_1$};
    \node[state] (q2) at (4,1) {$q_2$};
    \node[state] (q3) at (6,1) {$q_3$};
    \node (qi) at (7,1) {$\ldots$};
    \node[state, accepting] (qn) at (8,1) {$q_n$};

    \path (q0) edge [loop above] node {$0,1$} (q0);
    \path (q0) edge node {$1$} (q1);
    \path (q1) edge node {$0,1$} (q2);
    \path (q2) edge node {$0,1$} (q3);
\end{tikzpicture}
\caption{The NFA $\A_n'$, which recognises the language $\{w \in \{0, 1\}^{*} \mid \text { the $n^{\mbox{\scriptsize th}}$ to the last letter in $w$ is $1$}\}$.}
\label{fig:NFA2}
\end{figure}

Since $\lang(\G_n) = \lang(\D_n)$, we have that $\mathrm{PSem}_{\G_n}^{\M}(s_0) = \mathrm{PSem}_{\D_n}^{\M}(s_0)$. In addition, since $\D_n$ is deterministic, we have $\mathrm{PSem}_{\D_n}^{\M}(s_0) = \mathrm{PSyn}_{\D_n}^{\M}(s_0)$.
  
It remains to show $\mathrm{PSyn}_{\G_n}^{\M}(s_0) \ge \mathrm{PSyn}_{\D_n}^{\M}(s_0)$. For that, it suffices to show that for an arbitrary accepting run $r$ of $\M \times \D_n$, there is a deterministic strategy for $\M \times \G_n$ such that $r'$ (the corresponding run in the product) is accepting in $\M \times \G_n$ where the projections of $r$ and $r'$ on $\M$ are the same.

Consider an accepting run of $\M \times \D_n$. Before entering an accepting end-component of $\M \times \D_n$, $\G_n$ always stays in $q_0$.
Once an accepting end-component of $\M \times \D_n$ is entered,
since its projection on $\D_n$ is an SCC in $\D_n$,
there must exist a word $w \in \lang(\B_n)$ (and thus $w \in \lang(\A_n)$), which is a possible initial sequence of letters produced from some state $m$ of $\M \times \D_n$ in this end-component.
We fix such a word $w$; such a state $m$ of the end-component in $\M \times \D_n$ from which this word $w$ can be produced; and strategy of the NBW $\G_n$ to follow the word $w$ from $q_0$ to the accepting state $f$.

In an accepting end-component we can be in two modes: \emph{tracking}, or \emph{not tracking}.
If we are \emph{not tracking} and reach $m$, we start to \emph{track} $w$:
If we are not in $q_0$, we move to \emph{not tracking}.
Otherwise, we use $\G_n$ to reach an accepting state after reading $w$ (ignoring what happens in any other case) with the strategy we have fixed. 
Note that, when starting from $m$, it is always possible, with a fixed probability $p>0$, to read $w$ from $q_0$, and thus to accept.
If we read an unexpected letter (where the `expected' letter is always the next letter from $w$) or the end of the word $w$ is reached, we move to \emph{not tracking}.

In the \emph{not tracking} mode, the strategy for $\G_n$ is to always stay in $q_0$ when it is in $q_0$. 
Thus, $\G_n$ can always re-set to the $q_0$ state. 
We do not need to specify actions for the other states in $\G_n$ as they have only deterministic transitions.

Once in an accepting end-component of $\M \times \D_n$, tracking is almost surely started infinitely often, and it is thus almost surely successful infinitely often.
Thus, we have
$\mathrm{PSyn}_{\G_n}^{\M}(s_0)\geq \mathrm{PSyn}_{\D_n}^{\M}(s_0)$.
\end{proof}

\subsection{The Marking Algorithm}

We show that the DBA which recognises $\lang_n^{\omega}$ is at least $\frac{1}{n+2}$ times the size of the minimum DFA which recognises $\lang_n$.
To do that, we first present an algorithm which marks the states of the DBA with finite words from $\Sigma^{+}$. 

We assume that the input DBA is finite, complete and all states are reachable from the initial state. 
It is a reasonable assumption, as we can see from \cref{lem:dba-complete}, a DBA which recognises $\lang_n^{\omega}$ with all states reachable from the initial state is complete. 
All states of $\D_n$ are complete because $\lang_n^{\omega}$ is a tail objective, an objective whose occurrence is independent of any finite prefix.

\begin{lemma}\label{lem:dba-complete}
 Let $\D_n$ be a DBA that recognises $\lang_n^{\omega}$. Let all states in $\D_n$ are reachable from the initial state. The DBA $\D_n$ is complete.
\end{lemma}

\begin{proof}
Assume for contradiction that the state $q$ in $\D_n$ is not complete and is missing a $\sigma$-transition where $\sigma \in \Sigma = \{0, 1, \$\}$. 
Since all states are reachable, there is a finite word $w$ that takes the initial state to $q$. 
All infinite words that start with $w\sigma$ will be rejected by $\D_n$.
For any infinite word $w' \in \lang_n^{\omega}$ that begins with $\sigma$, the infinite word $ww'$ should also be in $\lang_n^{\omega}$. 
However, $ww'$ is rejected by $\D_n$; hence, the contradiction. 
\end{proof}

For an arbitrary $n \in \nat$, define the language 
$$\Gamma_n = \{w \in \{0, 1\}^{*}\$ \mid \text { the $n^{\mbox{\scriptsize th}}$ letter before $\$$ in $w$ is $0$}\}\ .$$
We have $\Gamma_n$ is a subset of the complement language of $\lang_n$, that is, $\Gamma_n \subseteq \Sigma^* \setminus \lang_n$.
Same as $\lang_n$, all words in $\Gamma_n$ are at least of length $n+1$.

Here comes the marking algorithm for a DBA $\D_n$ which recognises $\lang_n^{\omega}$:

\begin{algorithm}[ht]
\SetAlgoLined
\DontPrintSemicolon
\KwIn{A complete DBA $\D_n$ which recognises $\lang_n^{\omega}$ and an empty marking function $V: Q^\D \pfun \Sigma^+$}
\KwOut{A (partial) marking function $V: Q^\D \pfun \Sigma^+$}
\setstretch{1.35}
Mark all final states with $0$, that is, $V(q) = 0$ for all $q \in Q^\D_f$;\;
For every unmarked state $q$ from which a final state can be reached through a word $w \in \{0,1\}^+$, fix such a word $w_q$ and mark it with this word, that is, $V(q) = w_q$; \;
For every unmarked state $q$ from which a 
marked
state can be reached through a word $w \in \Gamma_n$, 
fix such a word $w_q$ and mark it with this word, that is, $V(q) = w_q$. \;
\caption{\mbox{Marking Algorithm for DBA}}
\label{alg:marking-DBA}
\end{algorithm}

Given an empty marking function and a finite, complete DBA in which all states are reachable from the initial state, the marking algorithm proceeds in three simple steps.
It first marks all final states with the letter $0$ on line~1. 
Next, it marks all states that can reach the final states by going through a finite word of only $0$s and $1$s with that word on line~2. 
Finally, it marks one by one the state that can reach any already marked state via a word in $\Gamma_n$ with that word.
As there are only finitely many states in the input DBA, the marking algorithm always terminates. 

\subsection{GfM More Succinct Than GfG}
In this section, we show that GfM NBA is exponentially more succinct than GfG NBA. 
We show that, in the marked DBA output by \cref{alg:marking-DBA}, there must be some states that are unmarked.
We produce a new DFA, $\mathcal P$, by collapsing all marked states into a single accepting sink without any outgoing transitions. 
In $\mathcal P$, the unmarked states and the transitions between them are the same as the original marked DBA. 
We then show that $\mathcal P$ is at least $\frac{1}{n+2}$ times the number of states as the DFA which recognises $\lang_n$.  

Let the input DBA be $\D_n = (\Sigma, Q^\D,q_0^\D, \delta^{\D}, Q_f^{\D})$. 
When the algorithm terminates, a state $q$ is unmarked if $V(q)$ is undefined. 
We first show an important property of the algorithm:


\begin{lemma}\label{lem:unmarked-state-trans}
After line~2 and after 
line~3,
we have for every unmarked state $q$ that $\delta^{\D}(q,0)$ and $\delta^{\D}(q,1)$ are unmarked.  
\end{lemma}

\begin{proof}
It is easy to see that, for a state $q$ that is unmarked after executing line~2, all finite word $w \in \{0, 1\}^+$ will take $q$ to unmarked states only. 
Otherwise, assume for contradiction that $q$ reaches a marked state $q'$ on a finite word $w \in \{0, 1\}^+$, we have two cases:
\begin{itemize}
    \item If $q'$ is final, $q$ should be marked on line~2.
    \item If $q'$ is marked with $V(q')$ on line~2, then $q$ will reach a final state on $wV(q')$, and $q$ should also be marked on line~2.  
\end{itemize}

Now we prove for a state $q$ which remains unmarked after 
line~3,
both $\delta^{\D}(q,0)$ and $\delta^{\D}(q,1)$ are unmarked. 
Assume for contradiction $q' = \delta^{\D}(q,0)$ is marked. 
We distinguish the following two cases.
\begin{itemize}
    \item 
    If $q'$ is marked after executing line~2, a final state is reached through the finite word $V(q') \in \{0, 1\}^+$ from $q'$. 
    Thus, a final state is reached through the finite word $0V(q') \in \{0, 1\}^+$ from $q$, which is a contradiction as $q$ should be marked on line~2.    
    \item
     Otherwise, $q'$ is unmarked after executing line~2, it is then marked 
     on line~3.
     A marked state is reached through the finite word $V(q') \in \Gamma_n$ from $q'$.
     Thus, the same marked state is reached through the finite word the finite word $0V(q') \in \Gamma_n$ from $q$, then $q$ should be marked on line~3, which again is a contradiction. 
\end{itemize}

Similarly, we can show by contradiction that $\delta^{\D}(q,1)$ is also unmarked. 
\end{proof}

When the algorithm terminates, we have a partial marking function such that a state of the DBA is either marked with a finite word $w \in \Sigma^+$ or unmarked.
\cref{lem:unmarked-state-trans} tells us that in the marked DBA, on any word from $\{0, 1\}^+$ any unmarked states will only go through unmarked states on the way.
We show in the next lemma that there must be some unmarked states.

\begin{lemma}\label{lem:DBA-not-all-marked}
There exist states $q$ such that $V(q)$ is undefined.   
\end{lemma}

\begin{proof}
Assume for contradiction that all states are marked, that is, $V(q) \in \Sigma^+$ for all $q \in Q^\D$. 
We successively build an infinite word which consists of a sequence of finite words. 
While we are building this infinite word, we also build the run of $\D_n$ on this word.
We start with the finite word $V(q_0^\D)$ where $q_0^\D$ is the initial state.
The selection of subsequent finite words is based on the end state that is reached in the run
sequence that $\D_n$ produces when traversing the input word $w$: 
for state $q$, we use $V(q)$ where  $q = \delta^\D (q_0^\D, w)$ and $w$ is the word we have built before adding $V(q)$.

It is not hard to see that the infinite word we have built does not belong to $\lang_n^\omega$ as it does not contain any finite word in $\lang_n$.
We show that the run on this word, however, is accepting in $\D_n$, which contradicts that $\D_n$ recognises $\lang_n^\omega$.

Let $w = V(q_0^\D)V(q_1^\D)V(q_2^\D)\ldots$ be the infinite word we have built. 
Let the prefix be $w_i = V(q_0^\D)\ldots V(q_i^\D)$ for $i \ge 0$ and $q_i^\D = \delta(q_0^\D, w_{i-1})$ for $i > 0$. 
The run of $\D_n$ on $w$ will visit all states $q_i^{\D}$ for $i \ge 0$. 
If there are infinitely many indices $i$ such that $q_i^{\D}$ is final, the run is accepting.
Otherwise, there are only finitely many indices $i$ such that $q_i^{\D}$ is final. We distinguish the following two cases:
\begin{itemize}
\item 
Assume that we have infinitely many indices $i$ such that $q_i^{\D}$ is marked on line~2 of \cref{alg:marking-DBA}. 
We have that $q_{i+1}^{\D}$ is final by running from $q_i^{\D}$ on the finite word $V(q_i^{\D})$,  which contradicts the assumption that we only have finitely many indices $i$ such that $q_i^{\D}$ is final.
\item 
Assume that starting from some positive index $i$, we only have states $q_i^{\D}$ that are marked on line~3 of \cref{alg:marking-DBA}. 
Since the input DBA is finite, there must exist indices $j$ and $k$ such that $i < j < k$ and $q_j^{\D} = q_k^{\D}$.
For each index $x$ such that $j \le x < k$, we have $q_{x+1}^{\D} = \delta\big(q_x^{\D}, V(q_x^{\D})\big)$ and $q_{x+1}^{\D}$ must be marked before $q_x^{\D}$ when running line~3 of \cref{alg:marking-DBA}.
This implies $q_{k}^{\D}$ is marked before $q_{j}^{\D}$, which leads to the desired contradiction as they are the same state.%
\vspace{-1.3em}
\end{itemize}
\end{proof}

After \cref{alg:marking-DBA} terminates, there must be some states that are marked, as all final states in $Q^{\D}_f$ are marked on line~1. 
The set of final states $Q^{\D}_f$ could not be empty, otherwise $\lang(\D_n) = \emptyset \neq \lang_n^\omega$.
There also must be some unmarked states by \cref{lem:DBA-not-all-marked}.
We produce a new automaton, $\mathcal P$, by collapsing all marked states into a single accepting sink without any outgoing transitions. 
We treat $\mathcal P$ as a DFA.

We show two important properties of $\mathcal P$.
First, we show a property that holds for all states in $\mathcal P$, that is, for all states $q$ in $\mathcal P$, $\mathcal P_q$ rejects all words in $\Gamma_n$.  

\begin{lemma}\label{lem:P-all}
For all states $q$ in the automaton $\mathcal P$ and all words $w$ in $\Gamma_n$, $\mathcal P_q$ rejects $w$.
\end{lemma}

\begin{proof}
Let $q$ be a state in the automaton $\mathcal P$.
Let $w$ be a word in $\Gamma_n$.
Towards contradiction, we assume $\mathcal P_q$ accepts $w$.
The sink state of $\mathcal P$ is reached through $w$ from $q$. 
Then, in the marked DBA $\D_n$, a marked state $q'$ is reached through $w$ from the unmarked state $q$.
The $q'$ state could be final (marked on line~1), or non-final (marked on line~2 or 3).
We have:
\begin{itemize}
\item 
Assume $q'$ is final. 
The state $q'$ should be marked on line~1. 
Since a final state is reached through $w \in \Gamma_n$, $q$ should be marked on line~3, which contradicts that $q$ is unmarked. 
\item 
Assume $q'$ is marked on line~2 or 3.
Since a marked state is reached through $w \in \Gamma_n$, $q$ should be marked on line~3, which contradicts that $q$ is unmarked.%
\vspace{-1.3em}
\end{itemize}
\end{proof}

Next, we show a property which holds for some state in $\mathcal P$, that is, there is a state $q$ in the automaton $\mathcal P$ such that $\mathcal P_q$ accepts all words in $\lang_n$.

\begin{lemma}\label{lem:P-exist-q}
There is a state $q$ in $\mathcal P$ such that for all $w \in \lang_n$, $\mathcal P_q$ accepts $w$. 
\end{lemma}

\begin{proof}
Towards contradiction, we assume that for all unmarked state $q$ in $\mathcal P$ (every state but the sink), there is a finite word $w_q \in \lang_n$ that is rejected by $\mathcal P_q$.
By \cref{lem:unmarked-state-trans} and the assumption $\mathcal P_q$ rejects $w_q$, we have that the automaton $\mathcal P_q$ not only traverses through unmarked states but also ends in an unmarked state.   
Using the words $w_q$ for all $q$ in $\mathcal P$, we build an infinite word. 
We show that this word is in $\lang_n$ but is rejected by $\D_n$. 
Similar to the proof of \cref{lem:DBA-not-all-marked}, we successively build this infinite word which consists of a sequence of finite words.

We start with $q_0^\D$. 
If $q_0^\D$ is marked in the DBA, we select $w_{q_0^\D}$ as the first finite word such that $\delta^{\D}(q_0^\D, w_{q_0^\D})$ is unmarked.
Such a word must exist, otherwise as $\D_n$ is complete, all states are marked, which contradicts \cref{lem:dba-complete}.
If $q_0^\D$ is not marked in the DBA and $w_{q_0^\D}$ is defined, we select $w_{q_0^\D}$ as the first finite word. 
The selection of subsequent finite words is based on the end state that is reached in the run sequence that $\D_n$ produces when traversing the input word $w$: 
for state $q$, we use $w_q$ where  $q = \delta^\D (q_0^\D, w)$ and $w$ is the word we have built before adding $w_q$.

The infinite word we have built belongs to $\lang_n^\omega$ as it contains infinite many finite words from $\lang_n$.
As we argued before, after reaching the first unmarked state, the run visits only unmarked states via the selected words from $\lang_n$. 
Thus, the run on the infinite word we built is rejecting in $\D_n$, as we visit no final states after we traverse the prefix $w_{q_0^\D}$, which contradicts that $\D_n$ recognises $\lang_n^\omega$.
\end{proof}

With \cref{lem:P-all} and \cref{lem:P-exist-q}, we have our main theorem which gives a lower bound of the size of $\mathcal P$.

\begin{theorem}\label{theorem:P-main}
The DFA $\mathcal P$ is at least $\frac{1}{n+2}$ times the size of the minimum DFA which recognises $\lang_n$.
\end{theorem}

\begin{proof}
According to \cref{lem:P-exist-q}, there exists a state $q$ such that $\mathcal P_q$ accepts all words in $\lang_n$.
By \cref{lem:P-all}, we also have $\mathcal P_q$ rejects all words in $\Gamma_n$.
Let $\C_n$ be the DFA which recognises the language $(0,1)^n(0,1)^*\$$. 
The DFA $\C_n$ (shown in Figure \ref{fig:DFA-goodWord}) has $n+2$ states.

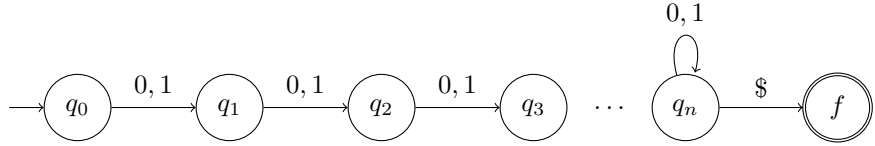
\begin{figure}[ht]
\centering
\begin{tikzpicture} [->, node distance = 2cm, auto,initial text = {}]
    \node[initial left,state] (q0) at (0,1) {$q_0$};
    \node[state] (q1) at (2,1) {$q_1$};
    \node[state] (q2) at (4,1) {$q_2$};
    \node[state] (q3) at (6,1) {$q_3$};
    \node (qi) at (7,1) {$\ldots$};
    \node[state] (qn) at (8,1) {$q_n$};
    \node[state, accepting] (qf) at (10,1) {$f$};

    \path (qn) edge [loop above] node {$0,1$} (qn);
    \path (q0) edge node {$0,1$} (q1);
    \path (q1) edge node {$0,1$} (q2);
    \path (q2) edge node {$0,1$} (q3);
    \path (qn) edge node {$\$$} (qf);
\end{tikzpicture}
\caption{The DFA $\mathcal C_n$ which recognises the language $(0,1)^n(0,1)^*\$$.}
\label{fig:DFA-goodWord}
\end{figure}

We have $\lang(P_q) \cap \lang(\C_n) = \lang_n$.
This is because: 
\begin{itemize}
\item Any word $w$ in $\lang_n$ is obviously in $\lang(\C_n)$, and in $\lang(\mathcal P_q)$ by \cref{lem:P-exist-q};
\item Any word $w$ in the intersection of $\lang(P_q)$ and $\lang(\C_n)$ must be in the language $\{0,1\}^+\$$ and have length at least $n+1$. 
Also, the $n^{\mbox{\scriptsize th}}$ letter before $\$$ in $w$ must be $1$. 
If the $n^{\mbox{\scriptsize th}}$ letter is $0$, it follows that the word is in $\Gamma_{n}$, 
which contradicts that $\mathcal P_q$ rejects all words in $\Gamma_{n}$.
Thus, $w \in \lang_n$.
\end{itemize}

We can construct a DFA which accepts $\lang(P_q) \cap \lang(\C_n)$, thus $\lang_n$, by intersecting $\mathcal P_q$ and $\C_n$.
The resulting DFA then has $n+2$ times as many states as $\mathcal P_q$, and thus at most $n+2$ times the number of states as $\mathcal P$.
Hence, $\mathcal P$ is at least $\frac{1}{n+2}$ times the size of the minimum DFA which recognises $\lang_n$.
   \end{proof}

The DFA which recognises $\lang_n$ should have at least $2^n$ states by \cref{theorem:DFA}. 
By \cref{theorem:P-main}, the automaton $\mathcal P$ should have at least $\frac{2^n}{n+2}$ states. 
This implies the simple corollary that $\D_n$ is at least as big as $\mathcal P$, and thus has at least $\frac{2^n}{n+2}$ states. 

\begin{corollary}\label{cor:Dn-exponential}
$\D_n$ has at least $\frac{2^n}{n+2}$ states.    
\end{corollary}

Noting that $\D_n$ can be chosen minimal, it is now easy to infer that GfM B\"uchi automata are exponentially more succinct than DBA by using the result that GfG B\"uchi automata are only up to quadratically more succinct than DBA for the same language \cite{KuperbergS15}.
The size of GfG automata is therefore in $\Omega(\sqrt{2^n/n})$ (or: $\Omega(2^{n/2}/\sqrt{n})$).


\begin{corollary}\label{cor:Gn-exponential}
GfM B\"uchi automata are exponentially more succinct than GfG B\"uchi automata.    
\end{corollary}

\section{GfM vs General Nondeterminism}
\label{section:GfMvsNondeterminism}
In this section, we show that general nondeterministic automata are exponentially more succinct than GfM automata.
As we can see in the following, this even holds when the languages are restricted to reachability or safety languages.
Using the same languages and general NBAs, we also observe unambiguous and separating automata are exponentially more succinct than GfM automata. Missing proofs can be found in \cref{appendix}.

The reachability languages we use are the extension of $\lang_n$ to reachability:
\[\lang_{n}^r = \{w \{0,1,\$\}^{\omega} \mid w \in \lang_n\}\ ,\]
the language of infinite words that contain at least one $\$$, as well as a $1$ that occurs exactly $n$ letters before this first $\$$.
In particular, words where the first $\$$ is one of the first $n$ letters is not part of $\lang_n^r$.



An NBA $\mathcal R_n$ that recognises $\lang_n^r$ is shown in \cref{fig:NRA}.
The automaton is a reachability automaton as there is only one accepting state, which is a sink. 

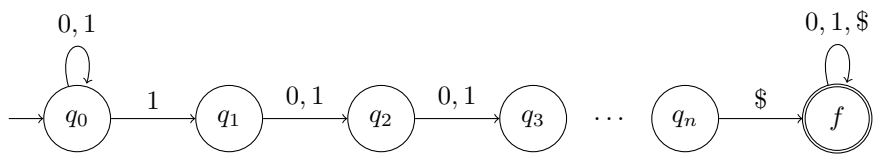
\begin{figure}[ht]
\centering
\begin{tikzpicture} [->, node distance = 2cm, auto,initial text = {}]
    \node[initial left,state] (q0) at (0,1) {$q_0$};
    \node[state] (q1) at (2,1) {$q_1$};
    \node[state] (q2) at (4,1) {$q_2$};
    \node[state] (q3) at (6,1) {$q_3$};
    \node (qi) at (7,1) {$\ldots$};
    \node[state] (qn) at (8,1) {$q_n$};
    \node[state, accepting] (qf) at (10,1) {$f$};

    \path (q0) edge [loop above] node {$0,1$} (q0);
    \path (q0) edge node {$1$} (q1);
    \path (q1) edge node {$0,1$} (q2);
    \path (q2) edge node {$0,1$} (q3);
    \path (qn) edge node {$\$$} (qf);
    \path (qf) edge [loop above] node {$0,1,\$$} (qf);

\end{tikzpicture}
\caption{The reachability NBA $\mathcal R_n$ that recognises $\lang_n^r$.}
\label{fig:NRA}
\end{figure}

The safety languages we use are the safety closure of an extension of words in $\lang_n$ by infintely many $\$$s: 
\[\lang_{n}^s = \{w \$^{\omega} \mid w \in \lang_n\} \cup  \{0, 1\}^{\omega}\ , \]

An NBA $\mathcal S_n$ that recognises $\lang_n^s$ is shown in \cref{fig:NSA}.
The automaton is a safety automaton as all states are accepting. 


\begin{figure}[ht]
\centering
\begin{tikzpicture} [->, node distance = 2cm, auto,initial text = {}]
    \node[initial left,state, accepting] (q0) at (0,1) {$q_0$};
    \node[state, accepting] (q1) at (2,1) {$q_1$};
    \node[state, accepting] (q2) at (4,1) {$q_2$};
    \node[state, accepting] (q3) at (6,1) {$q_3$};
    \node (qi) at (7,1) {$\ldots$};
    \node[state, accepting] (qn) at (8,1) {$q_n$};

    \path (q0) edge [loop above] node {$0,1$} (q0);
    \path (q0) edge node {$1$} (q1);
    \path (q1) edge node {$0,1$} (q2);
    \path (q2) edge node {$0,1$} (q3);
    \path (qn) edge [loop above] node {$\$$} (qn);

\end{tikzpicture}
\caption{The safety NBA $\mathcal S_n$ that recognises $\lang_n^s$.}
\label{fig:NSA}
\end{figure}
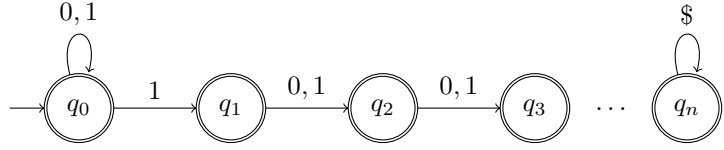

Next, we show that a GfM automaton for these languages needs to have $2^n$ states, not counting accepting or rejecting sinks.
For that, we construct a family of Markov chains as shown in \cref{fig:MC-big-GfM}.
We show that for an automaton, which recognises $\lang_n^r$ or $\lang_n^s$, to be good for this family of Markov chains, it needs at least $2^n$ states.

\begin{figure}[ht]
\centering
\begin{tikzpicture} [->, node distance = 2cm, auto,initial text = {}]
    \node[initial left,state] (q0) at (0,1) {$s_0$};
    \node[state] (q1) at (2,1) {$s_1$};
    \node[state] (q2) at (4,1) {$s_2$};
    \node[state] (q3) at (6,1) {$s_3$};
    \node (qi) at (7,1) {$\ldots$};
    \node[state] (qn) at (8,1) {$s_n$};
    \node[state] (qf) at (10,1) {$s_f$};

    \path (q0) edge node {$\textcolor{red}{1}:\sigma_1$} (q1);
    \path (q1) edge node {$\textcolor{red}{1}:\sigma_2$} (q2);
    \path (q2) edge node {$\textcolor{red}{1}:\sigma_3$} (q3);
    \path (qn) edge [loop above] node {$\textcolor{red}{\frac{1}{4}}: 0$} (qn);
    \path (qn) edge [loop below] node {$\textcolor{red}{\frac{1}{4}}: 1$} (qn);
    \path (qn) edge node {$\textcolor{red}{\frac{1}{2}}: \$$} (qf);
    \path (qf) edge [loop above] node {$\textcolor{red}{1}:\$$} (qf);

\end{tikzpicture}
\caption{A family of Markov chains, where $\sigma_1,\ldots,\sigma_n\in \{0,1\}$.
For a Markov chain with fixed values of $\sigma_1,\ldots,\sigma_n$,
the chance that it produces a word in $\lang_n^r$ (and $\lang_n^s$) is $2^{-(n+1)}+ \sum_{i=1}^{n} \sigma_i 2^{-i}$.
This is also the value that the product from the automaton state and $s_n$ has when $s_n$ is first reached.
Thus, $s_n$ must be reached with different states of a GfM automaton that recognises $\lang_n^r$ (or $\lang_n^s$) that resolves its nondeterminism optimally for different parameters $\sigma_1,\ldots,\sigma_n$, and there are $2^n$ combinations.}
\label{fig:MC-big-GfM}
\end{figure}
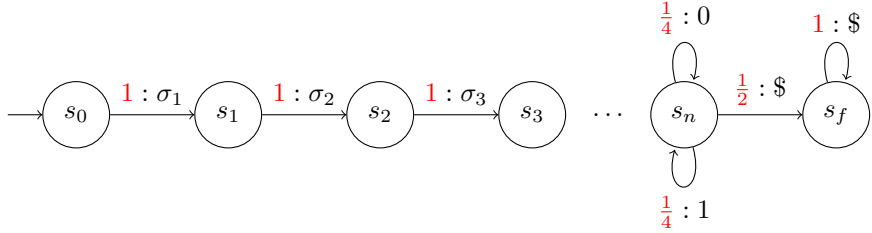

\begin{theorem}\label{theorem:GFM-exponential}
 A GfM automaton that recognises $\lang_n^r$ or $\lang_n^s$ has at least $2^n$ states.   
\end{theorem}

\begin{proof}
Let $n > 0$ and $\sigma_1,\ldots,\sigma_n \in \{0,1\}$. 
The family of Markov chains is shown in \cref{fig:MC-big-GfM}.
Since there are $2^n$ combinations of $\sigma_1,\ldots,\sigma_n$, we have $2^n$ Markov chains in this family. 

The chance that a Markov chain in this family produces a word in $\lang_n^r$ (and $\lang_n^s$) is $2^{-(n+1)}+ \sum_{i=1}^{n} \sigma_i 2^{-i}$.
We briefly explain how we obtain this value.
The chance that the first $\$$ is $i \leq n$ steps after reaching $s_n$ is $2^{-i}$, and the relative chance of winning then is $\sigma_i 2^{-i}$; this is the $\sum_{i=1}^{n} \sigma_i 2^{-i}$ part.
The chance that there never is a $\$$ is $0$.
The chance that there is a $\$$ after $>n$ steps after reaching $s_n$ is $2^{-n}$. The chance of winning in this case is $\frac{1}{2}$ (as a $0$ and $1$ exactly $n$ steps earlier are equally likely); this is the $2^{-(n+1)}$ part.

Let $\G$ be a GfM automaton that recognises $\lang_n^r$ (or $\lang_n^s$).
Now consider its product with one of the Markov chains $\M$. 
The chance that $\M$ produces a word in $\lang_n^r$ (and $\lang_n^s$) should also be the value of the state (a pair of $s_n$ and a state in $\G$) in the product $\M \times \G$ when $s_n$ is first reached.
Since the sub Markov chain that contain $s_n$ and $s_f$ is the same for all Markov chains in this family, 
in the product MDP $\M \times \G$,
the state (a pair of $s_n$ and a state in $\G$) when $s_n$ is first reached must be different. 
That is, $s_n$ must be paired with different states of $\G$ for different Markov chains in this family.
As there are $2^n$ different Markov chains,
there are at least $2^n$ states in $\G$.
\end{proof}

It can be seen from \cref{fig:NRA} resp.\ \cref{fig:NSA} that $\lang_n^r$ resp.\ $\lang_n^s$ are recognised by a nondeterministic reachability resp.\ safety automaton with $n+2$ resp.\ $n+1$ states. 
With \cref{theorem:GFM-exponential}, we have:
\begin{corollary}\label{cor:nba-succinct}
General nondeterministic automata are exponentially more succinct than GfM automata.    
\end{corollary}

We now consider the $\mathcal R_n$ and $\mathcal S_n$ more closely to see that we have used small subclasses of nondeterministic B\"uchi automata: the $\mathcal R_n$ are unambiguous reachability automata, while $\mathcal S_n$ is a strongly unambiguous separating safety automata. Let us start with the latter.

\begin{restatable}{lemma}{lemSnStrongUnambiguous}
For each $n\in \mathbb N$, $\mathcal S_n$ is a strongly unambiguous separating safety automaton.
\end{restatable}

With Theorem \ref{theorem:GFM-exponential}, this lemma provides:

\begin{corollary}\label{cor:separated-succinct}
Strongly unambiguous separating automata are exponentially more succinct than GfM automata.    
\end{corollary}

\begin{restatable}{lemma}{lemRnUnambiguous}
    For each $n\in \mathbb N$, $\mathcal R_n$ is an unambiguous reachability automaton.
\end{restatable}

As both $\mathcal R_n$ and $\mathcal S_n$ are unambiguous automata by the previous lemmas, Theorem \ref{theorem:GFM-exponential} provides:

\begin{corollary}\label{cor:unambiguous-succinct}
Unambiguous reachability and safety automata are exponentially more succinct than GfM automata.    
\end{corollary}

Finally, we can create a separating automaton with $n+2$ states that recognises the reachability language $\lang_n^r$, shown in Figure \ref{fig:SRA}. However, one cannot expect such automata to be (syntactic) reachability automata, because all words are accepted from an accepting sink state.

\begin{figure}[ht]
\centering
\begin{tikzpicture} [->, node distance = 2cm, auto,initial text = {}]
    \node[initial left,state] (q0) at (0,1) {$q_0$};
    \node[state] (q1) at (2,1) {$q_1$};
    \node[state] (q2) at (4,1) {$q_2$};
    \node[state] (q3) at (6,1) {$q_3$};
    \node (qi) at (7,1) {$\ldots$};
    \node[state, accepting] (qn) at (8,1) {$q_n$};
    \node[state, accepting] (qf) at (10,1) {$f$};

    \path (q0) edge [loop above] node {$0,1$} (q0);
    \path (q0) edge [below] node {$1$} (q1);
    \path (q1) edge [below] node {$0,1$} (q2);
    \path (q2) edge [below] node {$0,1$} (q3);
    \path (qn) edge node {$\$$} (qf);
    \path (qn) edge [bend right=40] node [near end,below] {$\$$} (q3);
    \path (qn) edge [bend right=40] node [near end,below] {$\$$} (q2);
    \path (qn) edge [bend right=40] node [near end,below] {$\$$} (q1);
    \path (qn) edge [bend right=40] node [near end,below] {$\$$} (q0);
    \path (qf) edge [loop above] node {$0,1$} (qf);
    \path (qn) edge [loop above] node {$\$$} (qn);
    \path (qf) edge [bend left=27] node [below] {$0$} (q1);

\end{tikzpicture}
\caption{The separating NBA $\mathcal R_n'$ recognises the reachability language $\lang_n^r$.}
\label{fig:SRA}
\end{figure}
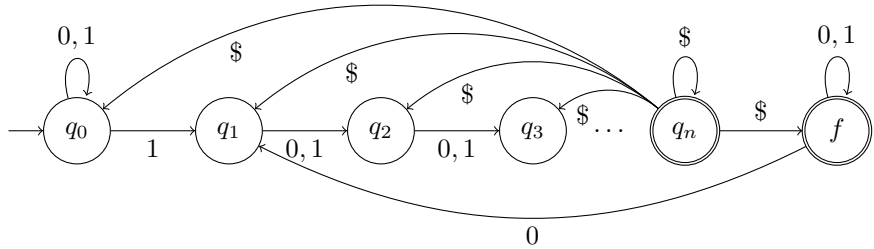

\begin{restatable}{lemma}{lemRnpSeparating}
\label{lem:Rn'-separating}
For each $n\in \mathbb N$, $\mathcal R_n'$ is a separating automaton that recognises $\lang_n^r$.
\end{restatable}

With Theorem \ref{theorem:GFM-exponential}, this lemma provides:

\begin{corollary}\label{cor:separated-reach-succinct}
Separating automata that recognise reachability languages are exponentially more succinct than GfM automata.    
\end{corollary}

\section{Conclusion}
We have established that GfM automata are exponentially more succinct than GfG automata, while general nondeterministic automata are exponentially more succinct than GfM automata. For the latter, we have shown that this is still the case when we restrict the class of general nondeterministic automata further, requiring them to be strongly unambiguous separating safety automata, unambiguous reachability automata, or separating automata that recognise a reachability language.

For all of these automata, we have provided very simple families of automata that witness the exponentially succinctness advantage.
\newpage\bibliography{bib}

\newpage\appendix\label{section:appendix}
\section{Proofs of \cref{section:GfMvsNondeterminism}}
\label{appendix}
\lemSnStrongUnambiguous*

\begin{proof}
It is obviously a safety automaton, as all states are final.

To see that $\mathcal S_n$ is separating, for all $i>0$ the language accepted from state $q_i$ is $\{0,1\}^{i-1}\$^\omega$,
while the language from state $q_0$ contains only words, whose first $n$ letters are in ${0,1}$ (as it would block on any word containing a $\$$ in its first $n$ letters).

To see that $\mathcal S_n$ is strongly unambiguous, we distinguish two types of words: those that do, and those that do not contain a $\$$. For the latter, $q_0^\omega$ is the only run, as all other runs that leave $q_0$ continue with $q_1,q_2,\ldots,q_n$ and then block.
For the former, as the automaton blocks when reading a $\$$ from every state except $q_n$, the only possible run is to move to $q_1$ exactly $n$ steps before the first $\$$ and then continue with $q_2,q_3,\ldots,q_{n-1},q_n^\omega$. (This will only be a run if the word is in $\lang_n^s$, but we only have to show that there is no other run.)
\end{proof}

\lemRnUnambiguous*

\begin{proof}
It is obviously a reachability automaton, as it has only a single accepting state, which is a sink.

To see that $\mathcal R_n$ is unambiguous, we look at a word in the language $\lang_n^r$ of the automaton. A word in $\lang_n^r$ has the form $\{0,1\}^m1\{0,1\}^{n-1}\}\$\{0,1,\$\}^\omega$, and its only accepting run is $q_0^{m+1},q_1,q_2,\ldots,q_n,f^\omega$:
any run that would move to $q_1$ earlier would block on reading a $0$ or $1$ while in $q_n$, while any run that would move to $q_1$ either later or not at all would block when reading a $\$$ in a state $q_i$ with $i<n$.
\end{proof}

\lemRnpSeparating*

\begin{proof}
To see that $\mathcal R_n'$ is separating, we note that
\begin{itemize}
    \item words starting with $\$$ can only be accepted from $q_n$, as the automaton blocks when reading $\$$ from any other state,
    \item words starting with $\{0,1\}^i\$$, with $0 < i < n$, can only be accepted from $q_{n-i}$, because
    \begin{itemize}
        \item when starting at $q_j$ with $j > n-i$, the automaton would block when reading a $0$ or $1$ after reaching $q_n$, while
        \item for all other states, the automaton would block when reading the first $\$$ while not having reached $q_n$,
    \end{itemize}
        \item words starting with $\{0,1\}^i1\{0,1\}^{n-1}\$$ can only be accepted from $q_0$, because
        \begin{itemize}
            \item from $q_i$ with $0<i \leq n$, the automaton would block when reading the $(n+1-i)$-th letter which is a $0$ or a $1$, while
            \item from $f$, it would block when moving upon reading the $j$-th letter with $j\leq i$ to $q_1$  when, reading the $(n+j)$-th letter (a $0$ or $1$) in state $q_n$ or, where $j>i+1$ (or where the automaton never moves to $q_1$), when reading the $(1+i+n)$-th letter, a $\$$, in a state in $q_k$ with $k<n$ (or $f$), 
        \end{itemize} 
        \item words starting with $\{0,1\}^i0\{0,1\}^{n-1}\$$ can only be accepted from $f$, using a similar argument, and 
        \item finally, words containing no $\$$ can only be accepted from $f$, because the only final state reachable from any other state of the automaton is $q_n$, and from $q_n$ the automaton would block, as it only accepts a $\$$ from there.
\end{itemize}
This partitions the set of words.
To see that $\mathcal R_n'$ recognises $\lang_n^r$, we note that, for all of the words that do contain a $\$$, this provides a recipe to be in $q_n$ when this $\$$ arrives from \emph{some} state, while for a word that does not contain a $\$$, we accept from $f$ by simply staying there for ever.
For a word in $\lang_n^r$, we can therefore follow this recipe until a $\$$ arrives, and then pick a successor state $q$ so that the tail of the word (after that $\$$) is in the language that can only be accepted from $q$, and so forth.
This constructs an accepting run.
(Note that we have already shown that words not in $\lang_n^r$ cannot be accepted from the only initial state, $q_0$.)
\end{proof}

\end{document}